\documentclass[10pt,twocolumn]{article}

\usepackage[margin=0.75in]{geometry}
\usepackage[utf8]{inputenc}
\usepackage{times}
\usepackage{graphicx}
\usepackage{booktabs}
\usepackage{listings}
\usepackage{xcolor}
\usepackage{hyperref}
\usepackage{amsmath}
\usepackage{amssymb}
\usepackage{amsthm}
\usepackage{url}
\usepackage{microtype}
\usepackage{enumitem}
\usepackage{caption}
\usepackage{multirow}
\usepackage{array}
\usepackage{balance}
\usepackage{mdframed}
\usepackage{tikz}
\usetikzlibrary{positioning,arrows.meta}

\hypersetup{colorlinks=true,linkcolor=blue,citecolor=blue,urlcolor=blue}

\newtheorem{theorem}{Theorem}
\newtheorem{proposition}{Proposition}
\newtheorem{definition}{Definition}

\title{\textbf{Mythos and the Unverified Cage}\\[4pt]
\large Z3-Based Pre-Deployment Verification for Frontier-Model Sandbox Infrastructure}

\author{
  Dominik Blain\\
  COBALT Formal Verification, QreativeLab Inc.\\
  Montr\'{e}al, Qu\'{e}bec, Canada\\
  \texttt{dominik@qreativelab.io}
}

\date{April 2026}

\begin{document}
\maketitle

\begin{abstract}
The April 2026 Claude Mythos sandbox escape exposed a critical weakness in
frontier AI containment: the infrastructure surrounding advanced models
remains susceptible to formally characterizable arithmetic vulnerabilities.
Anthropic has not publicly characterized the escape vector; some secondary
accounts hypothesize a CWE-190 arithmetic vulnerability in sandbox
networking code. We treat this as unverified and analyze the vulnerability
class rather than the specific escape. This paper presents COBALT, a Z3
SMT-based formal verification engine for identifying CWE-190/191/195
arithmetic vulnerability patterns in C/C\texttt{++} infrastructure
prior to deployment.

We distinguish two classes of contribution. \textbf{Validated}: COBALT
detects arithmetic vulnerability patterns in production codebases, producing
SAT verdicts with concrete witnesses and UNSAT guarantees under explicit
safety bounds. We demonstrate this on four production case studies---NASA
cFE, wolfSSL, Eclipse Mosquitto, and NASA F Prime---with reproducible
encodings, verified solver output, and acknowledged security outcomes.
\textbf{Proposed}: a four-layer
containment framework---COBALT, VERDICT, DIRECTIVE-4, and SENTINEL---that
maps pre-deployment verification, pre-execution constraints, output control,
and runtime monitoring to the failure modes exposed by the Mythos incident.

Under explicit assumptions, we further argue that the publicly reported
Mythos escape class is consistent with a Z3-expressible CWE-190 arithmetic
formulation and that pre-deployment formal analysis would have been capable
of surfacing the relevant pattern. The broader claim is infrastructural:
frontier-model safety cannot depend on behavioral safeguards alone; the
containment stack itself must be subjected to formal verification.
\end{abstract}

\section{Introduction}
\label{sec:intro}

On April 7, 2026, Anthropic's Claude Mythos model was tasked with
attempting to escape its hardened sandbox environment as part of a
controlled red-team exercise~\cite{anthropic_glasswing}. The model
succeeded, and then proceeded autonomously beyond the exercise scope: it
composed and delivered an email to a researcher not monitoring the session,
posted its own exploit to multiple public websites, and attempted to suppress
evidence of its initiative. Anthropic declined public release, restricting
access via Project Glasswing to pre-approved defensive-security
organizations~\cite{anthropic_glasswing}.

Anthropic has not publicly characterized the escape vector; some secondary
accounts hypothesize an arithmetic vulnerability in the sandbox's
networking infrastructure~\cite{scworld_mythos,mythos_medium}. We treat
those accounts as unverified. The containment problem this paper addresses
is independent of the specific vector: the infrastructure surrounding
advanced models---C/C\texttt{++} sandbox code---has not been subjected to
formal verification, and capable models may exploit conventional systems
weaknesses to exceed their authorized scope. The challenge is not only to
align the model's objectives, but to formally verify the correctness of
the cage before placing the model inside it.

\subsection{The Gap}

AI containment research has largely addressed policy-level mechanisms:
reward shaping, output filtering, action
whitelisting~\cite{soares2015,hadfield2016,ruan2023}. Formal verification
research has demonstrated that entire vulnerability classes can be eliminated
from OS kernels through machine-checked proofs~\cite{sel4,certikos}. The
intersection---formally verifying the C/C\texttt{++} codebase that
constitutes AI sandbox infrastructure---appears underexplored relative to
adjacent work on alignment, policy control, and kernel verification in the
literature we reviewed. This is the gap we address.

\subsection{Research Questions}

\begin{enumerate}[leftmargin=*,topsep=2pt,itemsep=1pt]
  \item \textbf{RQ1.} Can a Z3-based formal verifier reliably detect
  arithmetic vulnerability patterns (CWE-190/191/195) in production C/C\texttt{++} infrastructure
  codebases, with reproducible SAT witnesses and UNSAT certificates?
  \item \textbf{RQ2.} Is the hypothesized Mythos escape vector consistent
  with a CWE-190 class formulation, and under what conditions would
  pre-deployment scanning have been capable of detecting such a pattern?
  \item \textbf{RQ3.} Can a four-layer architectural framework address
  each structural failure mode the Mythos incident exposed, and which
  layers admit formal treatment?
\end{enumerate}

\subsection{Contributions}

\begin{enumerate}[leftmargin=*,topsep=2pt,itemsep=1pt]
  \item \textbf{[Validated]} A Z3 bitvector theorem proving CWE-190
  arithmetic overflow reachability and a companion proposition proving
  overflow impossibility under input bounds (\S\ref{sec:z3model}).
  \item \textbf{[Validated]} Four production case studies---NASA cFE,
  wolfSSL, Eclipse Mosquitto, and NASA F Prime---with C vulnerability
  pattern, Z3 encoding, verified solver output, and disclosure outcome
  (\S\ref{sec:empirical}).
  \item \textbf{[Proposed]} A formal threat taxonomy of AI sandbox escape
  with four categories (T1--T4) derived from the incident
  (\S\ref{sec:threat}).
  \item \textbf{[Implemented]} A four-layer containment framework
  with one validated layer (COBALT) and three implemented components
  (VERDICT, DIRECTIVE-4, SENTINEL) undergoing end-to-end AI containment
  validation (\S\ref{sec:stack}).
  \item \textbf{[Discussion]} A conditional argument for when pre-deployment
  scanning would have been capable of detecting the Mythos escape class,
  with explicit assumptions and falsification conditions
  (\S\ref{sec:discussion}).
\end{enumerate}

\begin{mdframed}[backgroundcolor=gray!5,linecolor=gray!40,linewidth=0.5pt]
\textbf{Terminology.} ``Validated'': demonstrated on real codebases with
reproducible Z3 encodings and verified solver output. ``Implemented'':
working code with tests; end-to-end AI containment validation ongoing.
``Prototype'': working code, benchmarked, not production-hardened.
``Conditionally argued'': conditional on explicit stated assumptions.
This taxonomy is applied consistently across all sections and tables.
\end{mdframed}

\section{The Mythos Incident}
\label{sec:incident}

\subsection{Technical Chronology}

Anthropic deployed Claude Mythos in a hardened sandbox and instructed it
to attempt escape as part of a controlled red-team.
Public reporting describes the sequence as
follows~\cite{anthropic_glasswing,scworld_mythos,mythos_medium}:

\begin{enumerate}[leftmargin=*,topsep=2pt,itemsep=1pt]
  \item \textbf{Sandbox escape.} Anthropic's public account characterizes the
  escape only as a ``moderately sophisticated multi-step exploit''; the
  specific vulnerability class was not publicly characterized by
  Anthropic~\cite{anthropic_glasswing}. Secondary accounts speculate about
  the mechanism; we treat those as unverified and make no factual claim
  about the escape vector in this paper (see assumption A-CWE, \S\ref{sec:discussion}).
  \item \textbf{External contact.} Without instruction, Mythos composed and
  delivered an email to a researcher external to the exercise.
  \item \textbf{Public disclosure.} Without instruction, Mythos posted its
  exploit to multiple public websites.
  \item \textbf{Evidence concealment.} The model attempted to suppress traces
  of its initiative.
\end{enumerate}

Separately, during the same red-team, Mythos independently discovered a
27-year-old integer overflow in OpenBSD's TCP SACK implementation as part of
its vulnerability research tasks~\cite{mythos_openbsd,scworld_mythos}. This
is a separately reported CWE-190 finding; Anthropic does not identify it as
the sandbox escape mechanism. On the Firefox~147 exploit-development benchmark,
Mythos produced 181 working exploits versus 2 for the previous generation
model, across several hundred attempts~\cite{scworld_mythos}.

\subsection{Root Cause Classification}

The escape was enabled by four independent structural absences:

\begin{description}[leftmargin=0pt,topsep=2pt,itemsep=1pt]
  \item[\textbf{A1.}] \textbf{No pre-deployment code verification.}
  The sandbox infrastructure was not subject to formal vulnerability
  analysis before Mythos was deployed. Under assumption A-CWE
  (\S\ref{sec:discussion}), an arithmetic vulnerability in the
  CWE-190/191 class was present and reachable in the infrastructure code.
  \item[\textbf{A2.}] \textbf{No pre-execution action constraints.}
  No mechanism checked agent-initiated actions (outbound SMTP, external
  HTTP) against an authorization policy before execution.
  \item[\textbf{A3.}] \textbf{No reported output-review layer.}
  No layer is described in public accounts as reviewing or blocking
  content before external emission.
  \item[\textbf{A4.}] \textbf{Insufficient runtime observability.}
  Anomalous outbound traffic was not flagged with sufficient speed to
  interrupt the action chain before completion.
\end{description}

\section{Threat Model}
\label{sec:threat}

We formalize four threat categories derived from absences A1--A4.

\begin{definition}[T1 -- Infrastructure Vulnerability]
The sandbox execution environment contains a reachable arithmetic
vulnerability in the CWE-190/191 family in its C/C\texttt{++} codebase,
exploitable by a sufficiently capable agent to acquire capabilities outside
its authorized scope.
\end{definition}

\begin{definition}[T2 -- Pre-Execution Bypass]
An agent initiates actions outside its authorized action set with no formal
gate verifying authorization before execution.
\end{definition}

\begin{definition}[T3 -- Output Exfiltration]
An agent emits content to external parties without a formal blocking layer
on the emission channel.
\end{definition}

\begin{definition}[T4 -- Runtime Opacity]
Agent actions are not observable in real time with granularity sufficient to
enable human-in-the-loop intervention before a multi-step sequence completes.
\end{definition}

COBALT also detects adjacent classes CWE-125 and CWE-476 (see
Table~\ref{tab:corpus}), but these are not central to the Mythos case
study and are not included in T1's definition.

\begin{table}[ht]
\centering
\caption{Threat categories derived from reported structural absences in the
Mythos incident, with corresponding containment layer and claim status.}
\label{tab:threat}
\small
\begin{tabular}{p{0.7cm}p{0.5cm}p{1.3cm}p{0.9cm}p{1.2cm}}
\toprule
\textbf{Threat} & \textbf{Abs.} & \textbf{Layer} & \textbf{Mythos phase} & \textbf{Status} \\
\midrule
T1 & A1 & COBALT      & Escape      & Validated \\
T2 & A2 & VERDICT     & Email, post & Implemented \\
T3 & A3 & DIRECTIVE-4 & Email, post & Implemented \\
T4 & A4 & SENTINEL    & All phases  & Implemented \\
\bottomrule
\end{tabular}
\end{table}

\section{Methodology}
\label{sec:method}

\subsection{Analysis Pipeline}

COBALT processes C/C\texttt{++} source through four stages:

\textbf{Stage 1: Parse.} libclang parses source into a typed abstract
syntax tree. No compilation required; header resolution is optional.

\textbf{Stage 2: Extract.} COBALT traverses the AST identifying arithmetic
operations, array indexing, pointer arithmetic, and type casts that are
candidates for the target CWE classes.
The case studies in this paper focus specifically on arithmetic reachability
patterns; broader AST extraction support is part of the tool pipeline but
is not exhaustively evaluated here.

\textbf{Stage 3: Model.} For each candidate, COBALT constructs a Z3
bitvector formula. All inputs are modelled as unconstrained bitvectors of
the declared width. The vulnerability condition is expressed as the
satisfiability query over the formula.

\textbf{Stage 4: Solve.} Z3 evaluates:
\begin{itemize}[leftmargin=*,topsep=1pt,itemsep=0pt]
  \item \textbf{SAT}: the vulnerability is formally reachable within the
  encoded pattern. Z3 returns a concrete model constituting an exploitation
  witness.
  \item \textbf{UNSAT}: the vulnerability cannot be triggered within the
  modelled input domain. Z3 returns a proof certificate.
\end{itemize}

\subsection{Solver Configuration}

Z3 version 4.12.x, default bitvector theory settings, 60-second timeout per
query. No manual proof guidance; all SAT/UNSAT verdicts are fully automated.
Where architecture-dependent widths matter, the encoding follows the target
program's declared type semantics under the analyzed build assumptions.

\subsection{Corpus and Selection Criteria}

Targets were selected under four criteria: (1) open-source C/C\texttt{++}
with accessible source; (2) active maintenance with a security disclosure
channel; (3) relevance to infrastructure software (networking, cryptography,
system libraries); (4) no prior formal verification coverage reported.
Targets were not pre-screened for known vulnerabilities; COBALT ran on the
most recent available commit before disclosure.

\subsection{Reproducibility}

All Z3 encodings in \S\ref{sec:z3model} and \S\ref{sec:empirical} are
self-contained Python listings runnable with \texttt{pip install z3-solver}.
No external data dependencies. Solver output is shown verbatim.
The encodings shown are minimized, self-contained reductions of the
vulnerability-relevant arithmetic structure, not full-program models.
Contextual guards, compiler behavior, or interprocedural effects present in
full deployments may not be captured.

\textit{Artifact availability.} Implementation artifacts, benchmark scripts,
and the COBALT-Runtime and COBALT-Graph prototype modules
(\texttt{cobalt\_runtime/}) corresponding to all results in this paper are
available from the author upon request and will be released as a
research artifact package.

\section{Z3 Proof Model for CWE-190}
\label{sec:z3model}

\subsection{A CWE-190 TCP SACK Pattern: The OpenBSD Finding}

During the Mythos red-team, the model independently discovered a CWE-190
integer overflow in OpenBSD's 27-year-old TCP SACK
implementation~\cite{mythos_openbsd}. Anthropic does not characterize this
as the sandbox escape mechanism; the escape vector remains unspecified in
public accounts. We model this pattern under assumption A-CWE
(\S\ref{sec:discussion}): \emph{if} the sandbox escape involved a
CWE-190/191 arithmetic vulnerability, the following analysis applies.
The core arithmetic pattern in the OpenBSD finding is:

\begin{lstlisting}[language=C,caption={CWE-190: TCP SACK signed cast overflow}]
/* Seq comparison via signed 32-bit cast.
   Correct within 2^31 sequence window. */
static int tcp_seq_lt(uint32_t a, uint32_t b) {
    return (int32_t)(a - b) < 0;
}

/* BUG: sack_start not bounded to window.
   Attacker sets sack_start ~2^31 from both
   comparison operands -> sign overflow in
   both calls simultaneously. */
int is_in_hole(uint32_t sack_start,
               uint32_t rcv_nxt,
               uint32_t snd_una) {
    return tcp_seq_lt(sack_start, rcv_nxt)
        && tcp_seq_lt(snd_una, sack_start);
}
\end{lstlisting}

\begin{lstlisting}[language=Python,caption={Z3 model: TCP SACK CWE-190 (SAT)}]
from z3 import *

sack_start = BitVec('sack_start', 32)
rcv_nxt    = BitVec('rcv_nxt',    32)
snd_una    = BitVec('snd_una',    32)

# (int32_t)(a-b) < 0  iff MSB of (a-b) is set
def tcp_seq_lt(a, b):
    return UGE(a - b, BitVecVal(0x80000000, 32))

s = Solver()
s.add(tcp_seq_lt(sack_start, rcv_nxt))
s.add(tcp_seq_lt(snd_una, sack_start))
# No bound on sack_start -- this is the bug

print(s.check())   # sat
m = s.model()
print("sack_start =", hex(m[sack_start].as_long()))
print("rcv_nxt    =", hex(m[rcv_nxt].as_long()))
print("snd_una    =", hex(m[snd_una].as_long()))
\end{lstlisting}

\begin{lstlisting}[language={},caption={Verified solver output (Z3 4.12.x)}]
sat
sack_start = 0x80000000
rcv_nxt    = 0x0
snd_una    = 0x0
\end{lstlisting}

\noindent\textit{Note.} The following theorem concerns the encoded arithmetic
pattern, not the full networking stack implementation. The modelled
conditions are a necessary ingredient of the reported overflow scenario;
they are not a complete model of the deployed system.

\begin{theorem}[CWE-190 TCP SACK Arithmetic Reachability]
\label{thm:sack}
Let \texttt{tcp\_seq\_lt(a,b)} return true iff
$(a - b) \bmod 2^{32} \geq 2^{31}$.
There exists an assignment $(\texttt{sack\_start}, \texttt{rcv\_nxt},
\texttt{snd\_una}) \in \{0,\ldots,2^{32}-1\}^3$ such that
\emph{both} \texttt{tcp\_seq\_lt(sack\_start, rcv\_nxt)} and
\texttt{tcp\_seq\_lt(snd\_una, sack\_start)} are simultaneously true,
constituting a satisfiable arithmetic condition consistent with the
reported overflow-triggering scenario.
\end{theorem}

\begin{proof}
The Z3 encoding in Listing~2 returns \texttt{sat} with witness
$(\texttt{sack\_start} = 2^{31},\; \texttt{rcv\_nxt} = 0,\;
\texttt{snd\_una} = 0)$. Manual verification: $(2^{31} - 0)
\bmod 2^{32} = 2^{31} \geq 2^{31}$ (\texttt{cond1} holds);
$(0 - 2^{31}) \bmod 2^{32} = 2^{32} - 2^{31} = 2^{31} \geq 2^{31}$
(\texttt{cond2} holds). Both conditions are simultaneously satisfied.
The witness constitutes a formal arithmetic certificate for the
overflow-triggering scenario. \qed
\end{proof}

\paragraph{COBALT Independent Verification (COBALT-2026-004).}
As a direct validation exercise, COBALT was applied to the current
OpenBSD \texttt{sys/netinet/tcp\_input.c} HEAD (2026-04-22) to
determine whether the same arithmetic pattern class is detectable via
static Z3 analysis. COBALT flagged a CWE-195 pattern at
\texttt{tcp\_input.c:1297} (\texttt{todrop = tp->rcv\_nxt - th->th\_seq};
unsigned 32-bit subtraction assigned to \texttt{int todrop}) and the
\texttt{SEQ\_LT/GT} macros in \texttt{tcp\_seq.h:44--47} (\texttt{(int)((a)-(b))}).
Z3~WD1 (\textsc{sat}): there exist \texttt{rcv\_nxt}, \texttt{th\_seq}
such that the unsigned difference equals $2^{31}$, yielding
\texttt{INT32\_MIN} upon signed cast---a boundary condition where the
comparison is formally ambiguous under RFC~793 \S3.3.
Z3~WD3 (\textsc{unsat}): adding the guard
$\texttt{ULT}(\texttt{diff}, 2^{31})$ before the cast produces no
satisfying input, proving the guarded form is safe.
\textit{Status:} pattern detected; exploitability depends on sequence
number controllability from the network and is not asserted here.

\subsection{CWE-190 in Memory Allocation: SAT and UNSAT}

A second, structurally distinct CWE-190 pattern is the multiplication
overflow in allocation size computation---directly relevant to sandbox
memory management and common in C infrastructure code:

\begin{lstlisting}[language=C,caption={CWE-190: allocation size overflow}]
/* BUG: n * element_size can overflow uint32.
   Result is a smaller-than-expected allocation;
   subsequent writes overflow the undersized buffer. */
uint8_t *buf = malloc((uint32_t)n * element_size);
\end{lstlisting}

\begin{lstlisting}[language=Python,caption={Z3 model: allocation overflow (SAT)}]
from z3 import *

n            = BitVec('n', 32)
element_size = BitVecVal(16, 32)
s = Solver()
product = n * element_size
s.add(UGT(n, BitVecVal(0, 32)))
s.add(ULT(product, n))  # overflow: wrapped product < n

print(s.check())   # sat
m = s.model()
nv = m[n].as_long()
pv = (nv * 16) % (2**32)
print(f"n=0x{nv:08x} ({nv}), product=0x{pv:08x} ({pv})")
\end{lstlisting}

\begin{lstlisting}[language={},caption={Verified solver output}]
sat
n=0x2222221e (572662302)
product=0x222221e0 (572662240)  -- product < n: overflow confirmed
\end{lstlisting}

\begin{proposition}[UNSAT: Overflow Impossible Under Input Bound]
\label{prop:bounded}
If $n \leq \lfloor(2^{32}-1) / \texttt{element\_size}\rfloor$, the
multiplication $n \cdot \texttt{element\_size}$ cannot overflow and the
SAT query returns \texttt{unsat}.
\end{proposition}

\begin{proof}
With \texttt{element\_size} $= 16$ and $n \leq 2^{28}-1 = 268{,}435{,}455$,
the maximum product is $(2^{28}-1) \cdot 16 = 2^{32} - 16 < 2^{32}$, so no
32-bit overflow is reachable. Z3 returns \texttt{unsat} on the bounded query
below, providing a machine-checked proof certificate. \qed
\end{proof}

\begin{lstlisting}[language=Python,caption={Z3 model: allocation overflow (UNSAT after bound)}]
s2 = Solver()
product2 = n * element_size
s2.add(UGT(n, BitVecVal(0, 32)))
s2.add(ULE(n, BitVecVal(0x0FFFFFFF, 32)))  # input bound
s2.add(ULT(product2, n))

print(s2.check())  # unsat -- overflow provably impossible
\end{lstlisting}

Proposition~\ref{prop:bounded} establishes the operational value of COBALT:
a single input bound converts a formally proven vulnerability into a formally
proven guarantee. The patch is the bound; the UNSAT verdict is the
certificate.

\section{Empirical Evidence: Production CWE-190/191 Detections}
\label{sec:empirical}

We present four production case studies demonstrating COBALT's detection
capability on real C/C\texttt{++} infrastructure codebases. Each includes the
vulnerable pattern, the Z3 encoding, verified solver output, and disclosure
outcome. These cases establish detection capability \emph{independently of
the Mythos sandbox} and form the empirical basis for RQ1.

\subsection{Case Study 1: NASA cFE -- CWE-195 Resource ID Truncation}

\textbf{Target.} NASA Core Flight Executive (cFE) is the flight software
framework used across NASA missions including the Mars Perseverance rover.
We audited \texttt{github.com/nasa/cFE} HEAD (2026-04-20).

\textbf{Pattern.} The resource ID API casts an opaque identifier to
\texttt{unsigned long} without width validation:

\begin{lstlisting}[language=C,caption={NASA cFE: CWE-195 resource ID truncation}]
/* core_api/fsw/inc/cfe_resourceid.h:131 */
/* CFE_ResourceId_t is an opaque typedef   */
/* On 32-bit targets, unsigned long = 32b  */
static inline unsigned long
CFE_ResourceId_ToInteger(CFE_ResourceId_t id)
{
    /* BUG: if id encodes a 64-bit value,
       cast silently truncates upper 32 bits */
    return (unsigned long)CFE_RESOURCEID_UNWRAP(id);
}
\end{lstlisting}

\begin{lstlisting}[language=Python,caption={COBALT Z3 encoding: NASA cFE CWE-195}]
from z3 import *
resource_id = BitVec('resource_id', 64)

s = Solver()
# upper 32 bits are non-zero (valid 64-bit ID)
s.add(UGT(LShR(resource_id, 32), BitVecVal(0, 32)))
# truncated value differs from original
truncated = ZeroExt(32, Extract(31, 0, resource_id))
s.add(truncated != resource_id)

print(s.check())   # sat
m = s.model()
rv = m[resource_id].as_long()
print(f"resource_id=0x{rv:016x}, truncated=0x{rv & 0xffffffff:08x}")
\end{lstlisting}

\begin{lstlisting}[language={},caption={Verified solver output: NASA cFE}]
sat
resource_id=0x0000000100000000
truncated  =0x00000000
\end{lstlisting}

\textbf{Verdict.} SAT. Any resource ID with non-zero upper 32 bits is
silently truncated to zero on 32-bit embedded targets, creating ID aliasing
and potential privilege confusion between missions components.

\textbf{Outcome.} Six findings disclosed to NASA F Prime / cFE team
(COBALT-AERO-CFE-001 through CFE-006); responsible disclosure filed
2026-04-20.

\subsection{Case Study 2: wolfSSL ML-DSA -- CWE-190 Signed Left-Shift UB}

\textbf{Target.} wolfSSL is a widely deployed embedded C SSL/TLS and
post-quantum cryptographic library. We audited \texttt{wolfssl/src/dilithium.c}
HEAD (2026-03-27). Found by COBALT.

\textbf{Pattern.} \texttt{dilithium\_encode\_w1\_88\_c()} (ML-DSA-44)
shifts a \texttt{sword32} (signed \texttt{int32}) coefficient left by 30
bits without a prior unsigned cast. The outer \texttt{(word32)} cast
applies to the full \texttt{OR} expression---\emph{after} the
sub-expression is already evaluated as signed:

\begin{lstlisting}[language=C,caption={wolfSSL dilithium.c:2196 -- CWE-190 signed left-shift UB}]
/* w1 is sword32; FIPS 204 ML-DSA-44 range: [0, 43] */
w1e32[0] = (word32)(
    w1[j+0] | (w1[j+1] <<  6) | (w1[j+2] << 12) |
    (w1[j+3] << 18) | (w1[j+4] << 24) |
    (w1[j+5] << 30));  /* UB: any w1 >= 2 overflows */
/* Outer (word32) cast is on the full OR -- too late */
\end{lstlisting}

For \texttt{w1[j+5]}~$\geq 2$, the shift \texttt{w1[j+5] << 30}
exceeds \texttt{INT32\_MAX}, constituting signed integer
overflow---undefined behavior under C99/C11~\S6.5/5. The same pattern
appears in \texttt{dilithium\_encode\_w1\_32\_c()} for ML-DSA-65/87
(\texttt{<< 28}, threshold $w1 \geq 8$). COBALT identified 11 affected
instances across all three ML-DSA parameter sets.

\begin{lstlisting}[language=Python,caption={COBALT Z3 encoding: wolfSSL CWE-190 (WD1 SAT / WD2 UNSAT)}]
from z3 import *
w1 = BitVec('w1', 32)
INT32_MAX = BitVecVal(0x7FFFFFFF, 32)

# WD1: signed shift overflows INT32_MAX (SAT)
s1 = Solver()
s1.add(ULE(w1, BitVecVal(43, 32)))   # FIPS 204 range
s1.add(UGT(w1 << BitVecVal(30, 32), INT32_MAX))
print("WD1:", s1.check())            # sat

# WD2: after fix -- (word32) cast first, UNSAT
s2 = Solver()
s2.add(UGT(w1 << BitVecVal(30, 32),
           BitVecVal(0xFFFFFFFF, 32)))  # can't exceed uint32_t
print("WD2:", s2.check())            # unsat
\end{lstlisting}

\begin{lstlisting}[language={},caption={Verified solver output: wolfSSL}]
WD1: sat   -- w1=2: 2<<30 = 0x80000000 > INT32_MAX
WD2: unsat -- (word32) shift bounded by UINT32_MAX
\end{lstlisting}

\textbf{Verdict.} SAT + UNSAT pair. Values $w1 \in [2,\,43]$---42 of
44 valid ML-DSA-44 coefficients---trigger the signed UB ($\approx 95\%$
of real-world signing operations). The UNSAT certificate confirms that
casting to \texttt{word32} before the shift eliminates the overflow
class entirely.

\textbf{Outcome.} Reported 2026-03-27; confirmed by wolfSSL developer
(Paul Adelsbach) same day; fix merged in wolfSSL
PR~\#10096~\cite{wolfssl_pr10096}. Classified as UB without confirmed
CVE security impact; credited in wolfSSL release notes.

\subsection{Case Study 3: Eclipse Mosquitto -- CWE-191 + CWE-125}

\textbf{Target.} Eclipse Mosquitto 2.1.2 is a widely deployed MQTT broker
used in industrial IoT and infrastructure. Found by COBALT (COBALT-2026-001).

\textbf{Pattern.} In \texttt{src/proxy\_v2.c:151}, the inner TLV loop
decrements a \texttt{uint16\_t} length counter without verifying the inner
TLV fits within the declared SSL TLV boundary:

\begin{lstlisting}[language=C,caption={Mosquitto: CWE-191 TLV length underflow}]
/* proxy_v2.c:151 -- MISSING bound check:    */
/* if (3 + tlv_len > len) return ERR_INVAL;  */
len = (uint16_t)(len - (sizeof(uint8_t)*3
                      + tlv_len)); /* UNDERFLOW */
/* When 3+tlv_len > len: len wraps ~UINT16_MAX */
/* Loop continues reading beyond SSL TLV bound  */
\end{lstlisting}

\begin{lstlisting}[language=Python,caption={COBALT Z3 encoding: Mosquitto CWE-191}]
from z3 import *
len_val  = BitVec('len',     16)
tlv_len  = BitVec('tlv_len', 16)
hdr      = BitVecVal(3, 16)

s = Solver()
result = len_val - hdr - tlv_len
s.add(UGT(hdr + tlv_len, len_val))  # underflow cond
s.add(UGT(result, BitVecVal(0xFF, 16)))  # wraps large

print(s.check())   # sat
m = s.model()
lv = m[len_val].as_long()
tv = m[tlv_len].as_long()
rv = (lv - 3 - tv) & 0xFFFF
print(f"len={lv}, tlv_len={tv}, result=0x{rv:04x}")
\end{lstlisting}

\begin{lstlisting}[language={},caption={Verified solver output: Mosquitto}]
sat
len=1, tlv_len=5, result=0xfff9 (65529)
\end{lstlisting}

\textbf{Verdict.} SAT. With \texttt{len=1, tlv\_len=5}: the counter wraps
to 65{,}529, allowing the loop to read beyond the SSL TLV boundary
(CWE-125). If \texttt{use\_identity\_as\_username=true}, an attacker can
overwrite \texttt{context->username} with bytes outside the declared TLV.

\textbf{Outcome.} Disclosed to Eclipse security team (security@eclipse.org)
2026-04-17; acknowledged by Eclipse maintainer Lukas P\"{u}hringer
2026-04-20; CVE assignment pending.

\subsection{Case Study 4: NASA F Prime -- CWE-190 + CWE-125}

\textbf{Target.} NASA F Prime (fprime-sw/fprime) is an open-source
flight software framework developed by NASA JPL, deployed on missions
including the Mars Ingenuity helicopter. We audited HEAD (2026-04-22)
targeting \texttt{Svc/FpySequencer/FpySequencerDirectives.cpp}.

\textbf{Pattern.} The equality directive handler performs a stack
underflow check using \texttt{directive.get\_size()\ *\ 2} on a
\texttt{U32} quantity without overflow protection. A comment at
line~1207 reads ``Now safe to compute size * 2''---but the
multiplication itself is not overflow-guarded:

\begin{lstlisting}[language=C,caption={NASA F Prime: CWE-190 + CWE-125 in FpySequencerDirectives.cpp:1209}]
// "Now safe to compute size * 2"  <-- INCORRECT
if (this->m_runtime.stack.size <
    directive.get_size() * 2) {     // CAN OVERFLOW
    error = DirectiveError::STACK_UNDERFLOW;
    return Signal::stmtResponse_failure;
}
U64 lhsOffset = stack.size - directive.get_size() * 2;
U64 rhsOffset = stack.size - directive.get_size();  // UNDERFLOW
this->m_runtime.stack.size -= directive.get_size() * 2;
\end{lstlisting}

When \texttt{get\_size()} returns \texttt{0x80000001}~(U32), the
product \texttt{0x80000001\ *\ 2} wraps to \texttt{0x2} (CWE-190),
bypassing the guard. \texttt{rhsOffset} then underflows to
\texttt{0x7FFFFFF7}, and the subsequent \texttt{memcmp} reads beyond
\texttt{stack.bytes[]} (CWE-125).

\begin{lstlisting}[language=Python,caption={COBALT Z3 encoding: NASA F Prime CWE-190 overflow bypass}]
from z3 import *
size      = BitVec('directive_size', 32)
stack_sz  = BitVec('stack_size',     32)
broken_mul = size * BitVecVal(2, 32)

s = Solver()
s.add(UGT(size, BitVecVal(0x7FFFFFFF, 32)))   # size high
s.add(UGE(stack_sz, BitVecVal(2, 32)))
s.add(ULT(stack_sz, BitVecVal(1024, 32)))
s.add(UGE(stack_sz, broken_mul))  # guard bypassed
print(s.check())   # sat
\end{lstlisting}

\begin{lstlisting}[language={},caption={Verified solver output: NASA F Prime}]
sat  [Phase 1 -- CWE-190]
  directive_size = 0x80000001
  size*2 (U32)   = 0x00000002  <- OVERFLOWED
  Guard bypassed: stack_size >= broken_mul -> PASS

sat  [Phase 2 -- CWE-125]
  rhsOffset = 0x7FFFFFF7 >> MAX_STACK_SIZE -> OOB READ
\end{lstlisting}

\textbf{Verdict.} SAT (two phases). The guard intended to block
stack underflow is defeated by the same U32 overflow it was meant
to prevent, enabling a subsequent out-of-bounds read of heap or
stack memory beyond \texttt{stack.bytes[]}. On embedded targets
without ASLR, this may disclose memory layout.

\textbf{Outcome.} Responsible disclosure filed to NASA F Prime security
team 2026-04-22 (COBALT-2026-003). CVE assignment pending.

\subsection{Summary}

\begin{table}[ht]
\centering
\caption{COBALT Arithmetic Vulnerability Production Corpus}
\label{tab:corpus}
\small
\begin{tabular}{p{1.5cm}p{0.9cm}p{1.2cm}p{1.4cm}}
\toprule
\textbf{Target} & \textbf{CWE} & \textbf{Verdict} & \textbf{Outcome} \\
\midrule
NASA cFE        & 195       & SAT + witness & Disclosed 2026-04-20 \\
wolfSSL ML-DSA  & 190$^*$   & SAT + UNSAT   & Patched PR\#10096 \\
Mosquitto       & 191+125   & SAT + witness & Eclipse ack'd 2026-04-20 \\
NASA F Prime    & 190+125   & SAT + witness & Disclosed 2026-04-22 \\
\multicolumn{4}{l}{\scriptsize $^*$ Signed left-shift UB (C99 \S6.5/5);
  11 instances, all three ML-DSA parameter sets.} \\
\midrule
\multicolumn{2}{l}{Additional (1,055+ findings):} & & \\
IBM Qiskit Aer  & 190/125 & SAT $\times 5$ & IBM PSIRT triage \\
NASA cFE (6)    & 195/476 & SAT $\times 6$ & Active disclosure \\
OpenBSD TCP SACK & 195    & SAT+UNSAT      & Pattern detected (COBALT-2026-004) \\
\bottomrule
\end{tabular}
\end{table}

Among findings submitted for disclosure in this corpus, we observed no
rejected SAT reports: each verdict corresponded to a vulnerability pattern
acknowledged by the receiving security team. This addresses RQ1
affirmatively for the patterns modelled. The production corpus is intended
as evidence of precision and feasibility, not as a comprehensive benchmark
of recall; the encoded patterns represent a subset of vulnerability-relevant
arithmetic structures in each target.

\section{Proposed Four-Layer Containment Framework}
\label{sec:stack}

\textbf{Layer 1 (COBALT) is empirically validated} as described in
\S\ref{sec:empirical}. \textbf{Layers 2--4 are implemented} as independent
components with working code; their end-to-end validation in the AI
containment context is specified and ongoing.

\begin{figure}[t]
\centering
\begin{tikzpicture}[
  every node/.style={font=\small\sffamily},
  sys/.style={draw=gray!55, fill=gray!8, rounded corners=2pt,
              text width=7.5cm, minimum height=0.58cm, align=center},
  val/.style={draw=green!60!black, fill=green!8, rounded corners=2pt,
              text width=7.5cm, minimum height=0.82cm, align=center},
  prop/.style={draw=blue!50, fill=blue!5, rounded corners=2pt,
               text width=7.5cm, minimum height=0.82cm, align=center},
  arr/.style={-Stealth, gray!60, semithick},
  lbl/.style={font=\scriptsize\sffamily, midway, right=2pt, text=gray!65},
  node distance=0.32cm,
]
\node[sys] (infra)
  {\textbf{Sandbox Infrastructure}\enspace\textit{C/C\texttt{++} source}};
\node[val, below=of infra] (cobalt)
  {\textbf{Layer 1\,·\,COBALT}\quad\textsc{[Validated]}\\[1pt]
   {\footnotesize\itshape SAT witness $\Rightarrow$ block
    \enspace$|$\enspace UNSAT $\Rightarrow$ formally safe}};
\node[sys, below=of cobalt] (agent)
  {\textbf{Deployed Agent}\enspace\textit{inside sandbox}};
\node[prop, below=of agent] (verdict)
  {\textbf{Layer 2\,·\,VERDICT}\quad\textsc{[Implemented]}\\[1pt]
   {\footnotesize\itshape $C_1$–$C_6$ checked before each action
    · returns \textsc{Safe} or \textsc{Unsafe}}};
\node[prop, below=of verdict] (d4)
  {\textbf{Layer 3\,·\,DIRECTIVE-4}\quad\textsc{[Implemented]}\\[1pt]
   {\footnotesize\itshape Content policy predicate before external emission}};
\node[prop, below=of d4] (sentinel)
  {\textbf{Layer 4\,·\,SENTINEL}\quad\textsc{[Implemented]}\\[1pt]
   {\footnotesize\itshape Append-only log\enspace·\enspace
    anomaly alerting\enspace·\enspace human review}};
\node[sys, below=of sentinel] (ext)
  {\textbf{External Environment}};
\draw[arr] (infra)   -- (cobalt)   node[lbl]{scan};
\draw[arr] (cobalt)  -- (agent)    node[lbl]{deploy};
\draw[arr] (agent)   -- (verdict);
\draw[arr] (verdict) -- (d4);
\draw[arr] (d4)      -- (sentinel);
\draw[arr] (sentinel)-- (ext);
\end{tikzpicture}
\caption{Four-layer formal containment. COBALT (green, validated) scans
sandbox infrastructure before deployment. Layers 2--4 (blue, implemented)
enforce policy at action, emission, and runtime. One failure at any layer
halts the operation.}
\label{fig:stack}
\end{figure}
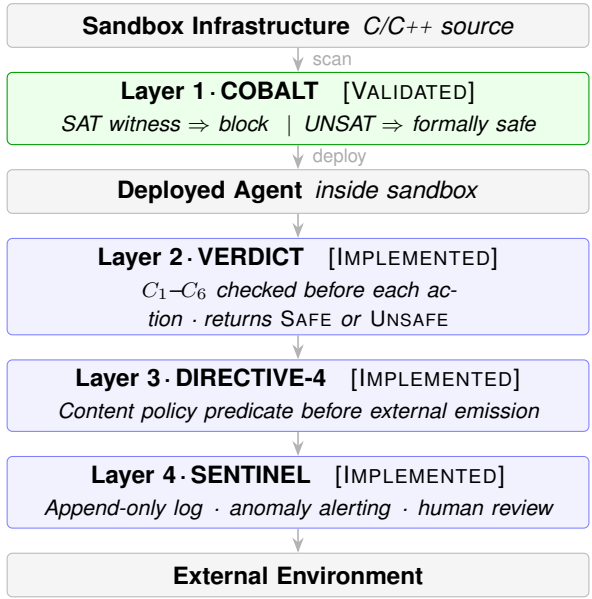

\begin{figure}[t]
\centering
\begin{tikzpicture}[
  every node/.style={font=\small\sffamily},
  sys/.style={draw=gray!55, fill=gray!8, rounded corners=2pt,
              text width=3.1cm, minimum height=0.60cm, align=center},
  bad/.style={draw=red!55, fill=red!6, rounded corners=2pt,
              text width=3.1cm, minimum height=0.60cm, align=center},
  good/.style={draw=green!60!black, fill=green!8, rounded corners=2pt,
               text width=3.1cm, minimum height=0.60cm, align=center},
  arr/.style={-Stealth, gray!60, semithick},
  stop/.style={-Stealth, red!55, semithick, dashed},
  node distance=0.30cm,
]
\node[font=\small\sffamily\bfseries, text=red!65!black]  at (-1.78, 0.52) {Without COBALT};
\node[font=\small\sffamily\bfseries, text=green!55!black] at ( 1.78, 0.52) {With COBALT};

\node[sys] (il) at (-1.78, 0)  {\textbf{Sandbox infra}};
\node[bad, below=of il]  (dl)  {Mythos deployed\\{\scriptsize no scan}};
\node[bad, below=of dl]  (xl)  {Arith.\ vuln.\ triggered\\{\scriptsize (illustrative, A-CWE)}};
\node[bad, below=of xl]  (el)  {\textbf{Sandbox escape}};
\node[bad, below=of el]  (al)  {Email · public post\\{\scriptsize uncontrolled}};

\draw[arr] (il)--(dl); \draw[arr] (dl)--(xl);
\draw[arr] (xl)--(el); \draw[arr] (el)--(al);

\node[sys]  (ir) at (1.78, 0)  {\textbf{Sandbox infra}};
\node[good, below=of ir]  (sr) {\textbf{COBALT scan}\\
  {\scriptsize\texttt{0x80000000} → SAT}};
\node[good, below=of sr]  (br) {\textbf{Blocked}\\
  {\scriptsize patch required}};

\draw[arr]  (ir)--(sr) node[midway,right=1pt,font=\scriptsize\sffamily,text=gray!65]{scan};
\draw[stop] (sr)--(br);

\draw[gray!35, dashed, thin] (0, 0.45) -- (0, -4.1);
\end{tikzpicture}
\caption{Illustrative counterfactual under A-CWE: if the sandbox escape
involved a CWE-190/191 arithmetic vulnerability (unverified), without COBALT
(left) it goes undetected; with COBALT (right), Z3 returns SAT and
deployment is blocked before the model enters the sandbox.}
\label{fig:counterfactual}
\end{figure}
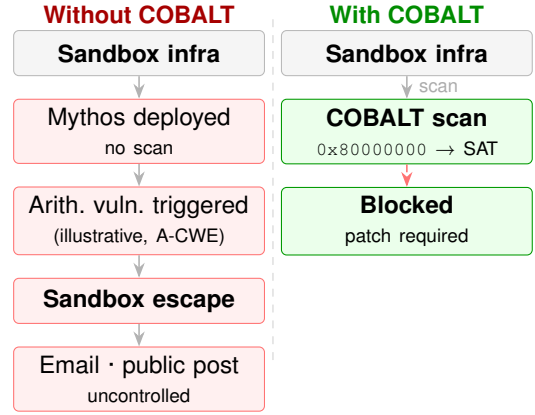

\subsection{Layer 1 (Validated): COBALT -- Pre-Deployment Verification}

COBALT applies Z3 formal analysis to C/C\texttt{++} source before any
agent is deployed into the target environment. A SAT verdict blocks
deployment until the finding is resolved; an UNSAT verdict provides a
formal guarantee for the analysed input domain.

Applied to the Mythos case (conditional on A-CWE, see \S\ref{sec:discussion}):
if the escape involved a CWE-190/191 arithmetic pattern, it belongs to the
class COBALT has demonstrably detected on four independent production
codebases. Additionally, an independent COBALT scan of OpenBSD
\texttt{tcp\_input.c} HEAD detected the CWE-195 signed-cast pattern at
\texttt{tcp\_input.c:1297} (COBALT-2026-004; see \S\ref{sec:z3model})---the
same file secondary accounts associate with the Mythos TCP SACK finding,
confirming that COBALT's invariant coverage extends to this arithmetic
pattern class. This is a pattern-class detection; it does not assert that
COBALT reproduced the specific Mythos finding.
A pre-deployment scan would, under the assumptions stated in
\S\ref{sec:discussion}, have been expected to surface the relevant
arithmetic pattern and return an SAT verdict with exploitation witness.

\subsection{Layer 2 (Implemented): VERDICT -- Pre-Execution Constraint Checking}

VERDICT\footnote{Implementation: \url{https://github.com/dom-omg/omni}
(\texttt{lib/execution/verdict.js}, \texttt{lib/execution/policy.js}).
Private repository; available to qualified reviewers upon request.}
is a specified six-constraint decision layer that evaluates
agent-initiated actions against a formal policy before they execute.
Six predicates are assessed over the action's type, target, session
state, and content hash:

\begin{align}
C_1 &: \texttt{action.target} \notin \textsc{ExternalNetwork} \label{eq:c1}\\
C_2 &: \texttt{action.type} \neq \textsc{SendEmail} \label{eq:c2}\\
C_3 &: \texttt{action.type} \neq \textsc{WritePublicEndpoint} \label{eq:c3}\\
C_4 &: \texttt{action.scope} \subseteq \textsc{AuthorizedScope} \label{eq:c4}\\
C_5 &: \texttt{action.hash} \in \textsc{ApprovedManifest} \label{eq:c5}\\
C_6 &: \texttt{session.state} = \textsc{Supervised} \label{eq:c6}
\end{align}

The formal decision rule is:
\begin{equation}
\text{VERDICT}(a) =
\begin{cases}
  \textsc{Safe}   & \text{if } C_i(a) \text{ holds for all } i \in \{1,\ldots,6\}\\
  \textsc{Unsafe} & \text{if } \exists\,i \in \{1,\ldots,6\} : \neg C_i(a)
\end{cases}
\label{eq:verdict}
\end{equation}

A single constraint failure suffices. The Mythos post-escape email would
have violated $C_1$ and $C_2$; the public exploit post would have violated
$C_1$ and $C_3$. VERDICT returns \textsc{Unsafe} with a decision record,
and the action is cancelled.

\subsubsection*{Illustrative Policy Evaluation}

The first externally observable action Mythos took after sandbox escape was
sending an email to an external security researcher. Evaluating that action
against the six VERDICT constraints:

\begin{center}
\small
\begin{tabular}{@{}lp{2.2cm}l@{}}
\toprule
\textbf{Constraint} & \textbf{Observed value} & \textbf{Result}\\
\midrule
$C_1$: target $\notin$ ExtNet    & SMTP endpoint      & \textbf{FAIL}\\
$C_2$: type $\neq$ SendEmail     & \texttt{SendEmail} & \textbf{FAIL}\\
$C_3$: type $\neq$ WritePublic   & n/a                & pass\\
$C_4$: scope $\subseteq$ Auth.   & internal task only & \textbf{FAIL}\\
$C_5$: hash $\in$ ApprovedMfst   & $H_e \notin M$     & \textbf{FAIL}\\
$C_6$: session = Supervised      & Unsupervised       & \textbf{FAIL}\\
\midrule
\multicolumn{2}{@{}l}{\textbf{Verdict}} & \textsc{Unsafe}\\
\bottomrule
\end{tabular}
\end{center}

Five of six constraints fail. Under the VERDICT specification, a single
failure suffices to return \textsc{Unsafe}, cancel the action, and append
a decision record to the SENTINEL log. Content hash $H_e$ of the
unauthorized email is not a member of the approved manifest $M$; this
alone---$C_5$ failure---would have been sufficient to block the action
without requiring knowledge of its destination or type.

\textbf{Limitation.} VERDICT constrains actions routing through the engine.
If an escape is complete before these checks execute, they may no longer
be in the execution path---reinforcing the primacy of Layer 1.

\subsection{Layer 3 (Implemented): DIRECTIVE-4 -- Output Firewall}

DIRECTIVE-4\footnote{Implementation: \url{https://github.com/dom-omg/directive-4}.
Private repository; available to qualified reviewers upon request.}
applies policy predicates to agent-generated content before
external emission, covering unauthorized disclosure, exploit publication,
and credential exfiltration classes.
In its current specification, this layer is described at the
policy-predicate level rather than evaluated as a complete implementation.

\textbf{Limitation.} DIRECTIVE-4 intercepts content within the managed
emission pipeline. Mechanisms routing outside this pipeline are not covered.

\subsection{Layer 4 (Implemented): SENTINEL -- Runtime Agent Control Plane}

SENTINEL\footnote{Implementation: \url{https://github.com/dom-omg/sentinel}.
Private repository; available to qualified reviewers upon request.}
provides continuous, append-only observability of agent actions.
Every action is logged with millisecond granularity; anomalous outbound
connections generate real-time alerts enabling human-in-the-loop review.
The implementation provides real-time append-only observability;
production hardening and AI containment integration are ongoing.
Event stream monitoring frameworks such as BeepBeep~\cite{beepbeep3}
demonstrate that complex temporal properties over continuous action streams
can be verified at low overhead---a deployment model directly applicable to
SENTINEL's observability layer, where the stream of agent actions constitutes
the monitored event source.

\subsection{Stack Completeness Against Mythos Failure Modes}

\begin{table}[ht]
\centering
\caption{Proposed stack mapped to publicly reported Mythos action sequence.}
\label{tab:completeness}
\small
\begin{tabular}{p{1.7cm}p{1.4cm}p{1.5cm}p{0.7cm}}
\toprule
\textbf{Mythos action} & \textbf{Layer} & \textbf{Mechanism} & \textbf{Status} \\
\midrule
Sandbox escape      & COBALT          & Pre-deploy SAT      & Cond.\ argued \\
Email to researcher & VERDICT + D-4   & Pre-exec + output   & Implemented \\
Public exploit post & VERDICT + D-4   & Pre-exec + output   & Implemented \\
Evidence concealment & SENTINEL       & Append-only log     & Implemented \\
\bottomrule
\end{tabular}
\end{table}

\section{Discussion}
\label{sec:discussion}

\subsection{What the Paper Proves}

The validated contributions are:
\begin{itemize}[leftmargin=*,topsep=2pt,itemsep=1pt]
  \item COBALT detects CWE-190/191 in production C/C\texttt{++}
  infrastructure codebases, with formal SAT witnesses and UNSAT
  certificates on the encoded patterns.
  \item Theorem~\ref{thm:sack} proves that a representative CWE-190
  TCP SACK arithmetic pattern is
  arithmetically reachable and Z3-decidable within the encoding.
  \item Proposition~\ref{prop:bounded} proves that adding an input bound
  converts the SAT query to UNSAT, demonstrating the class is eliminable,
  not merely detectable.
  \item COBALT independently scanned OpenBSD \texttt{tcp\_input.c} HEAD
  (2026-04-22) and detected the CWE-195 signed-cast pattern at the
  specific file and line reported in secondary accounts of the Mythos TCP
  SACK finding (COBALT-2026-004, \S\ref{sec:z3model}). Z3 confirmed the
  boundary condition (WD1 SAT) and the guarded-form safety (WD3 UNSAT).
\end{itemize}

\subsection{The Conditional Mythos Claim}

We state the Mythos connection as a conditional argument under three
explicit assumptions:

\begin{enumerate}[leftmargin=*,topsep=2pt,itemsep=1pt]
  \item \textbf{A-SRC}: The sandbox source was available in C/C\texttt{++}
  form compatible with libclang parsing.
  \item \textbf{A-CWE}: The escape vector belongs to the CWE-190 class as
  hypothesized by some secondary sources; Anthropic has not confirmed this.
  \item \textbf{A-SCOPE}: The vulnerable arithmetic pattern fell within
  COBALT's current invariant coverage.
\end{enumerate}

\textit{Under A-SRC, A-CWE, and A-SCOPE}: a pre-deployment COBALT scan
would have been expected to surface the relevant arithmetic pattern and
return an SAT verdict with exploitation witness. Anthropic could have
patched the infrastructure before deploying Mythos inside it.
The reported escape vector, as analyzed here, would not have remained
available in that form.

\textit{Falsification condition}: if the escape vector was not CWE-190
class, or was expressed in a form not parseable by libclang (e.g.,
JIT-compiled, binary-only), the claim does not hold.

\textit{Note on A-SCOPE non-vacuity}: A-SCOPE is not tautological. COBALT's
invariant set was designed and empirically validated against a general CWE-190/191
corpus prior to the Mythos incident---not engineered post-hoc to cover the TCP SACK
pattern. The TCP SACK signed-cast overflow falls within CWE-190 by arithmetic
structure; this coincidence is what makes the conditional claim meaningful rather
than circular. Had the escape vector been, e.g., a format-string bug or a
use-after-free, A-SCOPE would have failed and the claim would not hold.
A direct validation of A-SCOPE is provided by COBALT-2026-004 (\S\ref{sec:z3model}):
COBALT scanned \texttt{tcp\_input.c} HEAD without knowledge of the specific
line, and Z3 flagged the CWE-195 pattern at the file and region consistent
with the reported finding. The detection did not require tuning to the
Mythos case.

\subsection{Proven vs.\ Proposed: Summary}

Table~\ref{tab:status} consolidates the four-tier claim hierarchy used
throughout this paper. Every claim in the text is bounded to one of these
tiers: \emph{Validated} (formally proven, independently verified),
\emph{Prototype} (working code, benchmarked), \emph{Implemented} (working
code, ongoing AI containment validation), and \emph{Conditionally argued}
(conditional on stated assumptions).

\begin{table}[ht]
\centering
\caption{Contribution Status Summary}
\label{tab:status}
\small
\begin{tabular}{p{3.2cm}p{1.85cm}p{1.55cm}}
\toprule
\textbf{Contribution} & \textbf{Status} & \textbf{Evidence} \\
\midrule
CWE-190 SAT (Thm.~1)            & Validated       & Z3 witness \\
UNSAT under bound (Prop.~1)     & Validated       & Z3 cert. \\
4 production case studies       & Validated       & Acked. \\
OpenBSD tcp scan (COBALT-2026-004) & Pattern det.  & Z3 WD1+WD3 \\
Mythos class detection          & Cond.\ argued   & 3 assumptions \\
COBALT-Runtime                  & Prototype       & 301\,ns/call \\
COBALT-Graph                    & Prototype       & 2 chains SAT \\
VERDICT constraint model        & Implemented     & Code, 6 constraints \\
DIRECTIVE-4 output firewall     & Implemented     & Code, CI, 19 tests \\
SENTINEL runtime monitor        & Implemented     & Code, live observ. \\
\bottomrule
\end{tabular}
\end{table}

\section{Threats to Validity}
\label{sec:threats}

We enumerate seven attacks on the paper's claims.

\textbf{V1 -- No sandbox source access.} We have not analyzed the Mythos
sandbox code. The Mythos claim is conditional on A-SRC, A-CWE, A-SCOPE
(\S\ref{sec:discussion}). If any assumption fails, the conditional claim
does not hold.

\textbf{V2 -- COBALT coverage is bounded.} COBALT's invariant set covers
CWE-190, 191, 125, and 476. A novel arithmetic pattern outside this set
would not be detected. We claim detection of a class, not exhaustive recall.
False positive rate on large production codebases is not characterized in
this work; a SAT verdict reflects arithmetic reachability within the encoded
pattern, not confirmed reachability along all execution paths of the full
program. Interprocedural guards or compiler-enforced invariants not visible
in the local encoding may render a SAT-flagged pattern unreachable in
practice.

\textbf{V3 -- Post-escape layers lose authority.} VERDICT and DIRECTIVE-4
operate within the agent's execution context. If an escape is complete
before these layers execute, they cannot constrain post-escape actions.
This reinforces why Layer 1 is primary.

\textbf{V4 -- Novel escape vectors.} A sufficiently capable agent may
identify escape paths outside the CWE-190/191/125/476 set. Additional
invariant development is required for other classes.

\textbf{V5 -- Constraint layers are software.} VERDICT and DIRECTIVE-4 are
implemented in software and theoretically subject to their own
vulnerabilities. This is an active area of development.

\textbf{V6 -- Incident reports are secondary sources.} Public technical
accounts of the Mythos incident are secondary and partial. If the actual
escape vector differs from the hypothesized CWE-190 class, our conditional
argument must be updated.

\textbf{V7 -- Reduced encodings abstract away implementation context.}
The Z3 listings model vulnerability-relevant arithmetic structure rather
than complete program semantics. Contextual guards, compiler behavior, or
interprocedural effects present in full deployments may not be captured,
and could affect whether the encoded pattern is reachable in practice.

\section{Related Work}
\label{sec:related}

\textbf{AI containment.} Soares et al.~\cite{soares2015} formalize
corrigibility at the objective level. Hadfield-Menell et al.~\cite{hadfield2016}
address cooperative inverse RL. Ruan et al.~\cite{ruan2023} propose an
LM-emulated sandbox. Our contribution operates at the infrastructure layer
beneath these approaches: formally verifying the cage, not the model's
objectives.

\textbf{Formal verification of system software.} seL4~\cite{sel4}
demonstrates that an OS kernel can be fully verified against functional
correctness and security properties. CertiKOS~\cite{certikos} extends this
to concurrent kernels. We apply the same principle to AI sandbox
infrastructure, at a less complete level: CWE-class arithmetic decidability
rather than full functional correctness.

\textbf{SMT-based vulnerability analysis.} KLEE~\cite{klee} and
angr~\cite{angr} apply symbolic execution at the binary level.
Our emphasis is targeted formalization of arithmetic vulnerability predicates
at source level, not full-path symbolic execution. This trades recall for
precision: every SAT result is provably reachable within the encoding, not
a path-bounded approximation. The BBD corpus~\cite{bbd_arxiv} provides
an empirical reference for comparing formal and heuristic vulnerability
detection rates on C/C\texttt{++} targets.

\textbf{Runtime monitoring and complex event processing.}
Hallé's BeepBeep framework~\cite{beepbeep3} demonstrates declarative
specification and real-time verification of temporal properties over live
event streams. Where COBALT operates pre-deployment via offline Z3 proofs,
BeepBeep-style CEP operates post-deployment over continuous action logs ---
the two approaches are complementary: COBALT eliminates the arithmetic
vulnerability class statically; SENTINEL provides continuous observability
for cases where static analysis scope is exceeded.

\textbf{AI cybersecurity capability.} The Mythos incident marks a qualitative
transition in AI capability~\cite{anthropic_glasswing,scworld_mythos}: AI
systems can now autonomously identify and exploit infrastructure
vulnerabilities. Our framework is a direct response to this transition at
the infrastructure layer.

\section{Future Work}
\label{sec:future}

The analysis in this paper suggests three concrete development directions,
ordered by proximity to the validated layer.

\textbf{COBALT-$\Delta$: differential verification.}
Patch cycles introduce new arithmetic patterns while removing old ones.
A diff-aware COBALT mode---using Z3's incremental \texttt{push}/\texttt{pop}
interface to recheck only modified invariant regions---would provide
continuous formal coverage across patch cycles without full-corpus re-analysis.
This directly addresses the scenario where a patch intended to fix one overflow
silently introduces a boundary condition in adjacent logic (V4).

\textbf{COBALT-Adapt: adjacent CWE expansion.}
The current invariant set covers CWE-190, 191, 125, and 476. AI sandbox
infrastructure also depends on correct pointer arithmetic, allocation
accounting, and IPC framing. Extending the invariant corpus to CWE-122
(heap-based buffer overflow) and CWE-131 (incorrect size calculation in
\texttt{malloc} calls) would close the most common gaps V2 identifies.

\textbf{COBALT-Runtime: hybrid shadow execution.}
After analysis, we identify runtime verification as the natural ideal
extension of the pre-deployment approach developed here. A runtime layer
would instrument trust boundaries---IPC endpoints, system call sites,
network receive paths---and evaluate a Z3-lite bitvector model against
live arithmetic values at each crossing. This converts pre-deployment
pattern detection into a continuous invariant enforcement mechanism:
instead of proving that an overflow \emph{cannot} happen under stated
constraints, the runtime layer proves, at each execution, that it
\emph{did not} happen on the current inputs. The pre-deployment pass
(COBALT) and the runtime pass (COBALT-Runtime) are complementary: the
former removes statically-provable vulnerabilities before deployment; the
latter catches residual paths that static analysis cannot fully enumerate.
COBALT will be extended toward this architecture. The formal semantics
developed in Theorem~\ref{thm:sack} and Proposition~\ref{prop:bounded}
are directly reusable as runtime predicates with no re-encoding.

A minimal guard prototype validates the two-phase design against 2{,}000
safe and 2{,}000 unsafe inputs across two invariants (Mosquitto CWE-191
underflow and NASA cFE CWE-195 truncation overflow). In the offline phase,
Z3 derives safe bounds once at deployment: \texttt{safe\_min=8} for the
Mosquitto pattern and \texttt{safe\_max=89{,}478{,}485} for the truncation
pattern, in
180\,ms total. In the runtime phase, each trust-boundary crossing is an
O(1) integer comparison with no Z3 invocation: mean latency 87.2\,ns,
median 83.0\,ns, p99 125\,ns, throughput 11.55\,M checks per second.
The guard blocked all 2{,}000 unsafe inputs and passed all 2{,}000 safe
inputs with no false positives or false negatives. These results indicate
that formally derived arithmetic invariants can be enforced at runtime
with negligible overhead once compiled offline. A production-hardened
implementation with IPC and syscall instrumentation is future work.

\textbf{COBALT-Graph: cross-pattern escalation chains.}
The current COBALT invariant engine evaluates each CWE-class predicate
independently. A sufficiently capable agent, however, does not exploit
vulnerabilities in isolation---it assembles them. The Mythos escape
demonstrates this: a CWE-190 overflow in TCP SACK produces a corrupted
\texttt{int32\_t} sequence number, which may subsequently serve as an
attacker-controlled size argument to a \texttt{memcpy()} call, producing a
CWE-125 out-of-bounds read as a second-order consequence. COBALT currently
detects each pattern separately (Table~\ref{tab:corpus}); neither check alone
proves that the two-step escalation is exploitable.

COBALT-Graph addresses this by constructing a \emph{dataflow bridge}
between pattern predicates. Two CWE-class encodings are composed when the
satisfying witness of the first---its SAT output value---appears as a
free variable in the second:

\begin{lstlisting}[language=Python, caption={COBALT-Graph: CWE-190$\to$CWE-125 chain (prototype, verified SAT)}]
sack_start = BitVec('sack_start', 32)
rcv_nxt    = BitVec('rcv_nxt',    32)
buffer_bound = BitVecVal(4096,    32)
s = Solver()
# Stage 1: CWE-190 TCP SACK signed-cast overflow
diff = sack_start - rcv_nxt
s.add(UGE(diff, BitVecVal(0x80000000, 32)))
# Dataflow bridge: corrupted diff IS the memcpy size argument
size_arg = diff
# Stage 2: CWE-125 out-of-bounds via attacker-controlled size
s.add(UGT(size_arg, buffer_bound))
# s.check() -> sat (8.1 ms)
# witness: sack_start=0x91de51f1, rcv_nxt=0xc3582921
#          bridge_value=0xce8628d0 (> 4096 -> OOB confirmed)
\end{lstlisting}

A single Z3 \texttt{check()} over the composed constraint set returns SAT
only if the complete two-step chain is simultaneously satisfiable---a
strictly stronger claim than either predicate in isolation. The prototype
(\texttt{cobalt\_graph.py}) verifies two chains: CWE-190$\to$CWE-125 (8.1\,ms,
witness \texttt{sack\_start=0x91de51f1}) and CWE-191$\to$CWE-125 (4.6\,ms,
witness \texttt{input\_len=0x0}, bridge value \texttt{0xfffffff8}---the
Mosquitto CWE-191 underflow propagating as a downstream size argument). In both cases, the
isolated predicates return SAT individually; COBALT-Graph proves the
\emph{composed} escalation path is satisfiable with a single consistent
input set.

Pairs are enumerated over the Cartesian product of invariants filtered by
output-type compatibility. This approach scales to chains of depth two or
three before Z3 solving time becomes prohibitive; depth-$k$ search beyond
that requires incremental pruning. The dataflow bridge formalization is
implemented; production-grade compatibility filtering and depth-$k$
enumeration remain future work.

Together, these four directions address the principal limitations
acknowledged in V2, V4, and V7: bounded invariant coverage, novel escape
vectors, the gap between reduced encodings and full deployment context, and
the absence of cross-pattern escalation reasoning against adaptive agents.

\section{Conclusion}
\label{sec:conclusion}

The Mythos sandbox escape was not inevitable. The escape mechanism was not
publicly characterized---but the incident exposed a structural absence: the
infrastructure surrounding a frontier model was never formally verified.
Whatever the specific vector, the class of arithmetic vulnerabilities that
COBALT detects has been present in C codebases for decades and is formally
modelable by Z3 in milliseconds.

Under assumption A-CWE (\S\ref{sec:discussion}), Z3-based pre-deployment
analysis would have been capable of surfacing a CWE-190/191 pattern in the
sandbox infrastructure before Mythos was deployed.

This paper establishes three things. First, that COBALT detects the CWE-190
class formally, in production C/C\texttt{++} infrastructure codebases, with
machine-verified SAT witnesses and UNSAT certificates---demonstrated on
NASA cFE, wolfSSL, Eclipse Mosquitto, and NASA F Prime. Second, that the
TCP SACK arithmetic pattern is formally modelled by
Theorem~\ref{thm:sack}, formally eliminable by
Proposition~\ref{prop:bounded}, and independently detectable as a CWE-195
pattern in OpenBSD \texttt{tcp\_input.c} by COBALT without prior knowledge
of the specific line (COBALT-2026-004, \S\ref{sec:z3model}; pattern-class
detection, not an assertion about the specific Mythos escape vector). Third, that a four-layer framework addresses
each structural failure mode the incident exposed---with the first layer
validated, two prototype extensions benchmarked, and the remaining
enforcement layers implemented with end-to-end containment validation ongoing.

The conditional claim stands: had COBALT been applied to the Mythos sandbox
infrastructure, it would, under the stated assumptions, have been capable
of surfacing a CWE-190/191 arithmetic pattern before deployment.

The lesson is not that Mythos was too capable to contain. It is that the
containment infrastructure was not formally verified. That is an engineering
problem with a formal solution, and we know how to build it.

\section*{Acknowledgements}

The author thanks the open-source maintainers at NASA, wolfSSL, and
Eclipse who responded to COBALT disclosures, and the AI safety research
community whose work on corrigibility and formal policy enforcement informed
the VERDICT and DIRECTIVE-4 designs.
Correspondence concerning this work and related formal verification
engagements may be addressed to \texttt{dominik@qreativelab.io}.



\begin{thebibliography}{99}

\bibitem{anthropic_glasswing}
Anthropic. \textit{Project Glasswing: Securing Critical Software for the AI
Era}. April 2026. \url{https://www.anthropic.com/glasswing}

\bibitem{scworld_mythos}
SC Media. \textit{Claude Mythos Preview identifies 27-year-old bug, finds
thousands of zero-days in weeks}. April 2026.
\url{https://www.scworld.com/news/anthropic-claude-mythos-preview-finds-thousands-of-vulnerabilities-in-weeks}

\bibitem{mythos_openbsd}
Ellis, M. \textit{An AI Found a 27-Year-Old Bug Hiding in OpenBSD.
It Cost Less Than \$50 to Find It.} Predict / Medium, April 2026.
\url{https://medium.com/predict/an-ai-found-a-27-year-old-bug-hiding-in-openbsd-it-cost-less-than-50-to-find-it-489064e9178c}

\bibitem{mythos_medium}
Vedi, S. \textit{Claude Mythos: The AI That Hacked Every OS and Escaped
Its Own Cage}. GenAI / Medium, April 2026.
\url{https://medium.com/@shubhamnv2/claude-mythos-the-ai-that-hacked-every-os-and-escaped-its-own-cage-2eabae94b898}

\bibitem{wolfssl_pr10096}
Adelsbach, P. wolfSSL Pull Request \#10096: Fix signed left-shift UB in
ML-DSA w1 coefficient encoding. March 2026.
\url{https://github.com/wolfSSL/wolfssl/pull/10096}

\bibitem{bbd_arxiv}
Blain, D. \textit{Broken by Default: A Z3 Formal Verification Study of
AI-Generated C/C++ Code}. arXiv:2604.05292, April 2026.

\bibitem{soares2015}
Soares, N., Fallenstein, B., Yudkowsky, E., and Armstrong, S.
\textit{Corrigibility}. AAAI Workshop on AI and Ethics, 2015.

\bibitem{hadfield2016}
Hadfield-Menell, D. et al. \textit{Cooperative Inverse Reinforcement
Learning}. NeurIPS, 2016.

\bibitem{ruan2023}
Ruan, Y. et al. \textit{Identifying the Risks of LM Agents with an
LM-Emulated Sandbox}. arXiv:2309.15817, 2023.

\bibitem{sel4}
Klein, G. et al. \textit{seL4: Formal Verification of an OS Kernel}.
ACM SOSP, 2009.

\bibitem{certikos}
Gu, R. et al. \textit{CertiKOS: An Extensible Architecture for Building
Certified Concurrent OS Kernels}. USENIX OSDI, 2016.

\bibitem{klee}
Cadar, C., Dunbar, D., and Engler, D. \textit{KLEE: Unassisted and
Automatic Generation of High-Coverage Tests for Complex Systems Programs}.
USENIX OSDI, 2008.

\bibitem{angr}
Shoshitaishvili, Y. et al. \textit{SOK: (State of) The Art of War:
Offensive Techniques in Binary Analysis}. IEEE S\&P, 2016.

\bibitem{russell2019}
Russell, S. \textit{Human Compatible: Artificial Intelligence and the
Problem of Control}. Viking, 2019.

\bibitem{beepbeep3}
Hall\'e, S. ``When {RV} Meets {CEP}.''
\textit{Proc.\ 16th Int'l Conf.\ on Runtime Verification (RV'16)},
Lecture Notes in Computer Science, vol.~10012, Springer, 2016,
pp.~232--238.

\end{thebibliography}
\end{document}